\documentclass[12pt,a4]{article}

\usepackage{abstract}
\usepackage{amsfonts}
\usepackage{amsmath}
\usepackage{amssymb}
\usepackage{amsthm}
\usepackage{changepage}
\usepackage{setspace}
\usepackage[multiple]{footmisc}
\usepackage[round]{natbib}
\usepackage{graphicx}
\usepackage{subfig}
\usepackage{titlesec}

\newcommand{\mbA}{\mathbf{A}}
\newcommand{\mbbR}{\mathbb{R}}
\newcommand{\mbe}{\mathbf{e}}
\newcommand{\mbp}{\mathbf{p}}
\newcommand{\mbx}{\mathbf{x}}
\newcommand{\mby}{\mathbf{y}}
\newcommand{\mbzero}{\mathbf{0}}
\newcommand{\pref}{\succ}
\newcommand{\wpref}{\succcurlyeq}

\newtheorem{definition}{\normalfont\scshape Definition}
\newtheorem{proposition}{\normalfont\scshape Proposition}
\newtheorem{theorem}{\normalfont\scshape Theorem}

\renewcommand{\cal}{\mathcal}
\renewcommand{\epsilon}{\varepsilon}
\renewcommand{\geq}{\geqslant}
\renewcommand{\leq}{\leqslant}

\titleformat{\section}{\normalfont\scshape\centering}{\thesection}{1em}{}

\allowdisplaybreaks

\begin{document}

\title{Rationalizability, Cost-Rationalizability, \\ and Afriat's Efficiency Index\thanks{We are grateful to Roy Allen, Thomas Demuynck, Pawel Dziewulski, Federico Echenique, Georgios Gerasimou, Ludovic Renou, and Lanny Zrill for helpful comments.}}

\date{\normalsize{\today}}

\author{Matthew Polisson and John K.-H. Quah\thanks{Emails: matt.polisson@leicester.ac.uk; ecsqkhj@nus.edu.sg.}}

\maketitle

\begin{abstract} \linespread{1.1} \selectfont
This note explains the equivalence between approximate rationalizability and approximate cost-rationalizability within the context of consumer demand. In connection with these results, we interpret \citeauthor{afriat1973}'s (\citeyear{afriat1973}) critical cost efficiency index (CCEI) as a measure of approximate rationalizability through cost inefficiency, in the sense that an agent is spending more money than is required to achieve her utility targets.
\end{abstract}

\begin{abstract} \linespread{1.1} \selectfont
\end{abstract}

\begin{abstract} \linespread{1.1} \selectfont
\end{abstract}

\newpage

\section{Introduction} \label{section:introduction}

The notion of \emph{rationalizability} by utility maximization is natural and normatively uncontroversial in the revealed preference literature---a utility function is said to \emph{rationalize} expenditure data on different goods if, at each observation, it assigns weakly higher utility to the chosen bundle than to any other bundle that is weakly cheaper at the prevailing prices. But, in fact, there is an equally natural \emph{dual} notion which is just as important because it implicitly justifies many empirical studies of consumer demand and this is the notion of \emph{cost-rationalizability}---a utility function is said to \emph{cost-rationalize} expenditure data on different goods if, at each observation, any bundle that achieves at least as much utility as the observed bundle must be weakly more expensive.

In this note, we show that these two notions of rationalizabilty are \emph{definitionally equivalent} under mild conditions. We also show that they are \emph{observationally equivalent} in the sense that, in a data set with finitely many observations, both concepts are characterized by the generalized axiom of revealed preference (see \cite{afriat1967}, \cite{diewert1973}, and \cite{varian1982}). Our results can be seen as analogous to the famous \emph{duality} between Marshallian and Hicksian demands as solutions to the (primal) utility maximization and (dual) expenditure minimization problems, respectively.

In a sufficiently rich empirical setting, a data set will rarely be exactly rationalizable (equivalently, cost-rationalizable) and some way of measuring its departure from exact rationalizability is then required. The most common approach in the applied revealed preference literature is to calculate \citeauthor{afriat1973}'s (\citeyear{afriat1973}) \emph{critical cost efficiency index} (CCEI). While Afriat's measure is widely used in empirical applications, its motivation is not always well understood. Within the context of the duality between rationalizability and cost-rationalizability, we argue that the CCEI could be motivated as a measure of \emph{approximate} cost-rationalizability.

\section{Approximate Rationalizability} \label{section:rationalizability}

\noindent At observation $t$, the prices of $\ell$ goods are $\mbp^t$ and the consumer purchases the bundle $\mbx^t$ of these $\ell$ goods. We refer to the collection of observations, $\cal{O} = \{\mbp^t, \mbx^t\}_{t = 1}^T$, as a \emph{data set}. We assume throughout that $\mbp^t \in \mbbR^\ell_{++}$ and $\mbx^t \in \mbbR^\ell_+ \setminus \{\mbzero\}$.

\begin{definition} \label{definition:rationalizability}
\emph{Let $\cal{U}$ be a collection of utility functions defined on the consumption space $X \subseteq \mbbR^\ell_+$. The data set $\cal{O} = \{\mbp^t, \mbx^t\}_{t = 1}^T$ is \emph{rationalizable in $\cal{U}$ with respect to the vector of efficiency coefficients $\mbe = \{e^t\}_{t = 1}^T$} (where $e^t \in (0,1]$ for all $t$) (or simply $\mbe$-rationalizable in $\cal{U}$) if there is a utility function $U : X \to \mbbR$ belonging to $\cal{U}$, such that, at every observation $t$,
\[
U(\mbx^t) \geq U(\mbx) \; \; \mbox{for all $\mbx \in \cal{B}^t(e^t) = \{\mbx' \in X : \mbp^t \cdot \mbx' \leq e^t \mbp^t \cdot \mbx^t\}$.}
\]
A data set $\cal{O}$ is said to be \emph{exactly} rationalizable (or simply \emph{rationalizable}) in $\cal{U}$ if it is $(1, 1, \ldots, 1)$-rationalizable in $\cal{U}$.}
\end{definition}

Loosely speaking, the further is $\mbe$ from $(1, 1, \ldots, 1)$, the weaker is the rationalizability displayed within the data set. \cite{afriat1973} proposes the \emph{critical cost efficiency index} (CCEI) as a measure of a data set's distance from exact rationalizability.

\begin{definition} \label{definition:ccei}
\emph{Given a data set $\cal{O}$ and a family of utility functions $\cal{U}$, the CCEI is given by
\begin{equation} \label{equation:origin}
\sup \{e : \mbox{$\cal{O}$ is $(e, e, \ldots, e)$-rationalizable in $\cal{U}$}\}.
\end{equation}}
\end{definition}

To understand what the CCEI means, suppose that a data set $\cal{O}$ has a CCEI of 0.95 for the family of utility functions $\cal U$. Then the agent's observed choices are not exactly rationalizable in the sense that for \emph{any} utility function $U$ in $\cal{U}$, there is some observation $s$ and bundle $\mbx'$ such that $U(\mbx^s) < U(\mbx')$ and $\mbp^s \cdot \mbx' \leq \mbp^s \cdot \mbx^s$. However, this ``irrationality'' is limited in the sense that for any $\epsilon > 0$ sufficiently small, there is a utility function $\bar{U}$ in $\cal{U}$ such that $\bar{U}(\mbx^t) \geq \bar{U}(\mbx)$ for any $\mbx$ satisfying $\mbp^t \cdot \mbx \leq (0.95 - \epsilon) \mbp^t \cdot \mbx^t$ for all observations $t$; in other words, $\mbx^t$ gives weakly higher utility than any bundle that is more than 5\% cheaper.

\cite{afriat1973} provides a way to calculate the CCEI. To be precise, let $\cal{U}_{LNS}$ be the family of locally nonsatiated utility functions defined on  $X=\mbbR^{\ell}_+$.  It is shown that if $\cal{O}$ is $\mbe$-rationalizable by a member of $\cal{U}_{LNS}$, then $\cal{O}$ must obey a particular property that is easy to check; furthermore, if a data set satisfies this property then it is $\mbe$-rationalizable by a member of $\cal{U}_{WB}$, the family of \emph{well-behaved} utility functions, in the sense that the utility functions are defined on $X = \mbbR^{\ell}_{+}$, strictly increasing, and continuous. Given this result, a data set's CCEI can be obtained via a binary search over $\mbe = (e, e, \ldots, e)$, for $e \in (0,1]$, and the CCEI for the family $\cal{U}_{LNS}$ must coincide with that for the family $\cal{U}_{WB}$.

We now introduce the concepts needed to explain Afriat's result.

\begin{definition} \label{definition:relations}
\emph{Let $\cal{O} = \{\mbp^t, \mbx^t\}_{t = 1}^T$ be a data set and $\mbe = \{e^t\}_{t = 1}^T$ a vector of efficiency coefficients. For any pair of bundles $(\mbx^t, \mbx^s)$, the bundle $\mbx^t$ is \emph{(strictly) directly revealed preferred} to the bundle $\mbx^s$, denoted $\mbx^t \wpref_0^* (\pref_0^*) \; \mbx^s$, whenever $\mbp^t \cdot \mbx^s \leq (<) \; e^t \mbp^t \cdot \mbx^t$. For any finite sequence of bundles $(\mbx^t, \mbx^i, \mbx^j, \ldots, \mbx^l, \mbx^s)$, the bundle $\mbx^t$ is \emph{revealed preferred} to the bundle $\mbx^s$, denoted $\mbx^t \wpref^* \mbx^s$, whenever $\mbx^t \wpref_0^* \mbx^i$, $\mbx^i \wpref_0^* \mbx^j$, $\ldots$ , $\mbx^l \wpref_0^* \mbx^s$.}
\end{definition}

\begin{definition} \label{definition:garp}
\emph{Given a vector of efficiency coefficients $\mbe = \{e^t\}_{t = 1}^T$, a data set $\cal{O} = \{\mbp^t, \mbx^t\}_{t = 1}^T$ obeys $\mbe$-GARP (where GARP stands for the \emph{generalized axiom of revealed preference}) so long as $\mbx^t \wpref^* \mbx^s \implies \mbx^s \not\pref_0^* \mbx^t$. $\cal{O}$ satisfies GARP if it satisfies $\mbe$-GARP for $\mbe = (1, 1, \ldots, 1)$.\footnote{The term GARP is originally from \cite{varian1982}; \cite{afriat1967} introduces and refers to its equivalent as \emph{cyclical consistency}, and \cite{afriat1973} develops a modified version of this property (which is equivalent to $\mbe$-GARP with $\mbe = (e, e, \ldots, e)$).}}
\end{definition}

The following result is well known. For proofs, see \cite{afriat1973} and \cite{halevy2018}. We provide a further proof in the Appendix for the sake of completeness.

\begin{theorem} \label{theorem:afriat}
Let $\cal{O} = \{\mbp^t, \mbx^t\}_{t = 1}^T$ be a data set and $\mbe = \{e^t\}_{t = 1}^T$ a vector of efficiency coefficients. The following statements are equivalent:
\begin{itemize}
\item[(1)] $\cal{O}$ is $\mbe$-rationalizable in $\cal{U}_{LNS}$.
\item[(2)] $\cal{O}$ obeys $\mbe$-GARP.
\item[(3)] There is a set of numbers $\{\phi^t, \lambda^t\}_{t = 1}^T$ (with $\phi^t \in \mbbR$ and $\lambda^t \in \mbbR_{++}$), such that, at all $t, s$,
\[
\phi^s \leq \phi^t + \lambda^t \mbp^t \cdot (\mbx^s - e^t \mbx^t).
\]
\item[(4)] $\cal{O}$ is $\mbe$-rationalizable in $\cal{U}_{WB}$.\footnote{Recall that the families $\cal{U}_{LNS}$ and $\cal{U}_{WB}$ both contain utility functions which are defined on $X = \mbbR^\ell_+$. Also notice that statement (4) could be strengthened to say that $\cal{O}$ is $\mbe$-rationalizable in the family of well-behaved {\em and concave} utility functions.}
\end{itemize}
\end{theorem}

An immediate corollary of this result is that a data set's CCEI for the family of utility functions $\cal{U}_{LNS}$ coincides with that for the family $\cal{U}_{WB}$ and is given by
\begin{equation} \label{equation:carp}
\sup \{e : \mbox{$\cal{O}$ satisfies $(e, e, \ldots, e)$-GARP}\}.
\end{equation}

\section{Approximate Cost-Rationalizability} \label{section:cost-rationalizability}

\begin{definition} \label{definition:cost-rationalizability}
\emph{Let $\cal{U}$ be a collection of utility functions defined on the consumption space $X \subseteq \mbbR^\ell_+$. The data set $\cal{O} = \{\mbp^t, \mbx^t\}_{t = 1}^T$ is \emph{cost-rationalizable in $\cal{U}$ with respect to the vector of efficiency coefficients $\mbe = \{e^t\}_{t = 1}^T$} (or simply $\mbe$-cost-rationalizable in $\cal{U}$) if there is a utility function $U : X \to \mbbR$ belonging to $\cal{U}$, such that, at every observation $t$,
\[
e^t \mbp^t \cdot \mbx^t \leq \mbp^t \cdot \mbx \; \; \mbox{for all $\mbx \in \cal{P}^t = \{\mbx' \in X : U(\mbx') \geq U(\mbx^t)\}$.}
\]
A data set $\cal{O}$ is said to be \emph{exactly} cost-rationalizable (or simple \emph{cost-rationalizable}) in $\cal{U}$ if it is $(1, 1, \ldots, 1)$-cost-rationalizable in $\cal{U}$.}
\end{definition}

Exact cost-rationalizability means that there is some utility function $U$ in $\cal{U}$ for which the bundle chosen by the agent at each observation $t$ could be understood as the cheapest way of achieving a utility target $U(\mbx^t)$ given the prevailing prices $\mbp^t$.\footnote{See also the related notion of cost-rationalizability within the context of production in \cite{afriat1972}, \cite{hanoch1972}, and \cite{varian1984}. One important difference between production and consumption is that output levels are \emph{observed} in the former, while utility targets are \emph{unobserved} in the latter.}

The next proposition relates $\mbe$-rationalizability and $\mbe$-cost-rationalizability.

\begin{proposition} \label{proposition:equivalence}
\emph{(a)} If a locally nonsatiated utility function $U : \mbbR^\ell_+ \to \mbbR$ $\mbe$-rationalizes $\cal{O}$, then $U$ also $\mbe$-cost-rationalizes $\cal{O}$. \emph{(b)} If a continuous utility function $U : \mbbR^\ell_+ \to \mbbR$ $\mbe$-cost-rationalizes $\cal{O}$, then $U$ also $\mbe$-rationalizes $\cal{O}$.
\end{proposition}

\begin{proof}
If (a) is false then there is an observation $t$ and a bundle $\mby$ such that $U(\mby) \geq U(\mbx^t)$ with $e^t \mbp^t \cdot \mbx^t > \mbp^t \cdot \mby$. Since $U$ is locally nonsatiated, there is $\mby'$ such that $U(\mby') > U(\mby) \geq U(\mbx^t)$, with $e^t \mbp^t \cdot \mbx^t > \mbp^t \cdot \mby'$, which contradicts the $\mbe$-rationalizability of $\cal{O}$ by $U$.

For (b), we again prove by contradiction. Suppose that $U$ $\mbe$-cost-rationalizes $\cal{O}$ but that there is an observation $t$ and a bundle $\mby$ such that $U(\mby) > U(\mbx^t)$ with $\mbp^t \cdot \mby \leq e^t \mbp^t \cdot \mbx^t$. By the continuity of $U$ and the fact that $e^t \mbp^t \cdot \mbx^t > 0$, there is $\mby'$ such $\mbp^t \cdot \mby' < e^t \mbp^t \cdot \mbx^t$ and $U(\mby') > U(\mbx^t)$, but this is impossible since $U$ $\mbe$-cost-rationalizes $\cal{O}$.
\end{proof}

We know (from Theorem \ref{theorem:afriat}) that $\mbe$-GARP is sufficient for $\mbe$-rationalizability by a utility function $U$ belonging to $\cal{U}_{WB}$, and therefore (by Proposition \ref{proposition:equivalence}a) it is also sufficient for $\mbe$-cost-rationalizability by a utility function $U$ belonging to $\cal{U}_{WB}$. The necessity of $\mbe$-GARP is stated formally below.

\begin{proposition} \label{proposition:garp}
$\cal{O}$ obeys $\mbe$-GARP whenever any of the following conditions holds:
\begin{itemize}
\item[(i)] $\cal{O}$ is $\mbe$-rationalizable and $\mbe$-cost-rationalizable in $\cal{U}$, where $\cal{U}$ is a collection of utility functions $U : X \to \mbbR$ defined on the consumption space $X \subseteq \mbbR^\ell_+$.
\item[(ii)] $\cal O$ is $\mbe$-rationalizable by a locally nonsatiated utility function $U : \mbbR^\ell_+ \to \mbbR$.
\item[(iii)] $\cal{O}$ is $\mbe$-cost-rationalizable by a continuous utility function $U : \mbbR^\ell_+ \to \mbbR$.
\end{itemize}
\end{proposition}

\begin{proof}
In view of Propositon \ref{proposition:equivalence}, (ii) and (iii) follow once we prove (i). So we prove (i). Indeed, if $\mbp^t \cdot \mbx^s = e^t \mbp^t \cdot \mbx^t$, then $U(\mbx^s) \leq U(\mbx^t)$ since $U$ $\mbe$-rationalizes $\cal O$. Furthermore, if $\mbp^t \cdot \mbx^s < e^t \mbp^t \cdot \mbx^t$, then $U(\mbx^s) < U(\mbx^t)$ since $U$ is $\mbe$-cost-rationalizes $\cal{O}$. Therefore, if $\mbx^t \wpref^* \mbx^s$ then $U(\mbx^t) \geq U(\mbx^s)$ and so $\mbx^s \not\pref_0^* \mbx^t$, as required by $\mbe$-GARP (otherwise, we obtain $U(\mbx^s) > U(\mbx^t)$, which is a contradiction).
\end{proof}

\cite{pq2013} studies a family of utility functions $\cal U$ defined on a discrete consumption space $X$. Part (i) of Proposition \ref{proposition:garp} guarantees that GARP holds if a data set is rationalizable and cost-rationalizable by a utility function in $\cal U$; the authors show that GARP is also sufficient to guarantee the existence of a strictly increasing utility function defined on $X$ that rationalizes and cost-rationalizes the data set.

An implication of Proposition \ref{proposition:garp} is that, for the family of well-behaved utility functions $\cal{U}_{WB}$, $\mbe$-rationalizability and $\mbe$-cost-rationalizability are observationally equivalent because they are both equivalent to $\mbe$-GARP. Similarly, the CCEI (for $\cal{U} = \cal{U}_{WB}$) as defined by (\ref{equation:origin}), which is equal to (\ref{equation:carp}), is also equal to
\begin{equation} \label{equation:new}
\sup \{e : \mbox{$\cal{O}$ is $(e, e, \ldots, e)$-cost-rationalizable in $\cal{U}$}\}.
\end{equation}
This identity gives us a different way of interpreting the CCEI beyond its definition.

\section{Cost-Rationalizability and the CCEI} \label{section:ccei}

First, we argue that when analyzing \textbf{observational data}, it is reasonable to require both rationalizability and cost-rationalizability as rationality criteria. This is because, typically, the researcher would simply have access to just a fraction of the purchases made by the consumer. Even if the consumer did have an exogenously given budget, it is unlikely that all the goods on which that budget is spent are actually part of the researcher's observation.  For example, $\mbx^t$ could be the consumer's purchases of food products from a grocery store. The data set $\cal{O} = \{\mbp^t, \mbx^t\}_{t = 1}^T$ is cost-rationalizable if there is a utility function $U$ defined on those $\ell$ goods (properly speaking, a \emph{sub}-utility function, since her consumption decisions will cover many other (unobserved) goods) so that, at each observation $t$, the bundle $\mbx^t$ is the cheapest way of achieving utility level $U(\mbx^t)$. This is a reasonable rationality criterion to impose since we would expect the consumer to be cost efficient in her purchases of the $\ell$ goods, in order to maximize the money available for purchases of the non-$\ell$ goods.

In the context of cost-rationalizabilty, what does it mean for a consumer to have a CCEI of (say) 0.9? This means that for any $e' = 0.9 - \epsilon$ with $\epsilon > 0$ sufficiently small, $\cal{O}$ is cost-rationalizable with respect to $(e', e', \ldots, e')$. Then there is a well-behaved utility function $U$ defined on the $\ell$ goods, such that there is no observation $t$ at which the agent could have saved more than $100(0.1 + \epsilon)$\% of her observed expenditure $\mbp^t \cdot \mbx^t$ while achieving utility $U(\mbx^t)$. In other words, there could be one (or more) observation(s) at which she is spending more than is necessary given her utility target, but the cost saving is not materially greater than 10\% at any observation. Furthermore, given that 0.9 is the supremum among all possible values of $e$ at which $\cal{O}$ is $(e, e, \ldots, e)$-cost-rationalizable, for \emph{any} well-behaved utility function $U$, there must be at least one observation $s$ such that the agent could have achieved utility $U(\mbx^s)$ and saved more than $100(0.1 - \epsilon)$\% of her observed expenditure $\mbp^s \cdot \mbx^s$.

Now consider analyzing \textbf{laboratory data}, and suppose that a subject's data set is exactly rationalizable. Then the data set is also exactly cost-rationalizable, which means there is a utility function for which the subject will be strictly worse off should the experimenter give the subject a smaller budget (less of the experimental currency) in any one round. The CCEI is a measure of how close the subject is to exact cost-rationalizability. A CCEI of 0.9 means that, for a subject with \emph{any} well-behaved utility function, there will be one round (and possibly more) in which the subject could have achieved the same utility in that round with a budget that is $100(0.1 - \epsilon)$\% smaller. However, for any withdrawal of $100(0.1 + \epsilon)$\% of the budget, there is a well-behaved utility function such that a subject with such a utility function is strictly worse off in some round.

We should emphasize that in the definition of approximate cost-rationalizability, the target utility is contingent on the observation $t$ (or, if we interpret $t$ as time, contingent on time $t$). The definition does \emph{not} require that there is not a cheaper way of achieving the same \emph{set} of utility levels, possibly after altering the observation at which a particular utility level is obtained. For example, consider a data set where $\mbp^1 = (1, 1)$ and $\mbx^1 = (1, 1)$ at observation 1 and $\mbp^2 = (2, 2)$ and $\mbx^2 = (2, 2)$ at observation 2. This data set is rationalizable and cost-rationalizable (because it obeys GARP), but clearly the agent could save money if he had bought the bundle $\mbx^2 = (2, 2)$ at the prices $\mbp^1 = (1, 1)$ and the bundle $\mbx^1 = (1, 1)$ at the prices $\mbp^2 = (2, 2)$. This does not contradict cost-rationalizability because cost-rationalizability allows the agent to choose the timing of the utility targets. This greater permissiveness makes sense: observation 1 could be on food spending in normal times and observation 2 on food spending during a festive period; cost-rationalizability does not preclude the possibility of the agent choosing to derive higher utility from food during a festive period (even when food prices are higher) and lower utility from food in normal times (even when food prices are lower). In essence, cost-rationalizability respects the price-bundle pairs within observations, leaving open the possibility that prices/expenditures are in fact \emph{endogenous}.

\section{Discussion} \label{section:discussion}
 
Approximate rationalizability through cost inefficiency was first introduced within the context of production \citep{afriat1972} and then consumption \citep{afriat1973}.\footnote{\cite{debreu1951} also expresses a similar idea within the context of an exchange economy.} The intuition behind cost inefficiency in production is strong,\footnote{A cost inefficient combination of inputs means that the same level of output can be achieved through a different combination of inputs which costs strictly less at the prevailing factor prices.} and our objective in this note is to provide an equally strong intuition for cost inefficiency within the context of consumer demand. This in turn furnishes an intuitive interpretation for \citeauthor{afriat1973}'s (\citeyear{afriat1973}) critical cost efficiency index (CCEI) which may be useful for researchers in revealed preference. The CCEI is by far the most widely used index for measuring the degree of rationality in a data set and as such has received considerable attention, including some criticism (see, for example, \cite{echenique2022}). Whatever arguments there might be for and against, we do not see any interpretive difficulties with this index. Besides our interpretation, another (more ``psychological'') interpretation of the index based on a consumer's innate inability to distinguish among similar bundles is found in \cite{dziewulski2020}.

Our interpretation of the CCEI in terms of approximate cost-rationalizability remains valid when the index is applied to other classes of utility functions which are narrower than $\cal{U}_{WB}$. Note that one of the strengths of this index is the ease with which it could be calculated for other classes of utility functions. For example, in \cite{polisson2020}, CCEIs are calculated for a wide range of models of decision making within the context of choice under risk (e.g., expected utility ($\cal{U}_{EUT}$), disappointment aversion ($\cal{U}_{DA}$), rank dependent utility ($\cal{U}_{RDU}$), and stochastically monotone utility ($\cal{U}_{FOSD}$)) with and without concavity of the Bernoulli function.\footnote{The approach could also be applied straightforwardly to many prominent models of decision making under uncertainty and over time, which are formally very similar to models of decision making under risk (see \cite{polisson2020}).} It could also be estimated when $\cal{U}$ is a fully parametric class of utility functions, e.g., expected utility with constant relative risk aversion (CRRA).

We note that one could interpret the CCEI as just one way of \emph{aggregating} the efficiency coefficients $\mbe$ into a rationality index or score. Other prominent aggregators of $\mbe$ include the \cite{varian1990} and \cite{houtman1985} indices; see the discussions in \cite{halevy2018} and \cite{polisson2020}. The CCEI is the one most commonly used in empirical work,\footnote{See, for just a few examples, \cite{harbaugh2001}, \cite{andreoni2002}, \cite{choi2007a,choi2014}, \cite{fisman2007}, and \cite{carvalho2016}.} primarily because it is easy to calculate in a practical sense.\footnote{The \cite{afriat1973} index can be calculated efficiently using a binary search algorithm for a wide class of utility models (see \cite{polisson2020} and \cite{dembo2024}), whereas the \cite{varian1990} and \cite{houtman1985} indices are known to be more computationally demanding, and especially when $\cal U$ is not simply ${\cal U}_{WB}$.} There are also other ways of measuring departures from rationality that could not be thought of as simply different ways of aggregating $\mbe$; this includes \cite{echenique2020,echenique2023} and \cite{declippel2023} (which are based on first-order conditions), \cite{echenique2011} (which leverages the idea of a ``money pump''), and \cite{hu2021} (which is based on the size of perturbations of observed demand bundles).

\section*{Appendix} \label{section:appendix}

\noindent \emph{Proof of Theorem 1.} \, $(1) \implies (2)$: The $\mbe$-rationalizability of $\cal{O}$ by $U$ in $\cal{U}_{LNS}$ guarantees that $\mbp^t \cdot \mbx^s \leq e^t \mbp^t \cdot \mbx^t \implies U(\mbx^t) \geq U(\mbx^s)$ and $\mbp^t \cdot \mbx^s < e^t \mbp^t \cdot \mbx^t \implies U(\mbx^t) > U(\mbx^s)$. The first implication follows directly from the definition of $\mbe$-rationalizability of $\cal{O}$ by $U$, and the second from the definition of $\mbe$-rationalizability plus the local nonsatiation of $U$.\footnote{To see the latter, suppose that $\mbp^t \cdot \mbx^s < e^t \mbp^t \cdot \mbx^t \implies U(\mbx^t) = U(\mbx^s)$. (From the $\mbe$-rationalizability of $\cal{O}$ by $U$, $\mbp^t \cdot \mbx^s < e^t \mbp^t \cdot \mbx^t \implies U(\mbx^t) \geq U(\mbx^s)$.) Then by the local nonsatiation of $U$, there must be some bundle $\mby$ such that $\mbp^t \cdot \mby < e^t \mbp^t \cdot \mbx^t$ and $U(\mby) > U(\mbx^s) = U(\mbx^t)$, contradicting $\mbe$-rationalizability.} Therefore, for any $\mbp^t \cdot \mbx^i \leq e^t \mbp^t \cdot \mbx^t$, $\mbp^i \cdot \mbx^j \leq e^i \mbp^i \cdot \mbx^i$, $\ldots$ , $\mbp^l \cdot \mbx^t \leq e^l \mbp^l \cdot \mbx^l$, we obtain
\[
U(\mbx^t) \geq U(\mbx^i) \geq U(\mbx^j) \geq \cdots \geq U(\mbx^l) \geq U(\mbx^t),
\]
which of course implies that $U(\mbx^t) = U(\mbx^i) = U(\mbx^j) = \cdots = U(\mbx^l) = U(\mbx^t)$. This equality then equires that $\mbp^t \cdot \mbx^i = e^t \mbp^t \cdot \mbx^t$, $\mbp^i \cdot \mbx^j = e^i \mbp^i \cdot \mbx^i$, $\ldots$ , $\mbp^l \cdot \mbx^t = e^l \mbp^l \cdot \mbx^l$. In other words, for any revealed preference cycle $\mbx^t \wpref_0^* \mbx^i$, $\mbx^i \wpref_0^* \mbx^j$, $\ldots$ , $\mbx^l \wpref_0^* \mbx^t$, the weak relation $\wpref_0^*$ cannot be replaced with the strict relation $\pref_0^*$ anywhere in the cycle.

$(2) \implies (3)$: Define $a^{t, s} = \mbp^t \cdot (\mbx^s - e^t \mbx^t)$ for all $t, s$, and let $\mbA$ be a $T \times T$ square matrix with $a^{t, s}$ the $(t, s)$-th entry. Since $\cal{O}$ satisfies $\mbe$-GARP, $\mbA$ satisfies cyclical consistency in the sense defined in \cite{forges2009}, which guarantees that there is a set of numbers $\{\phi^t, \lambda^t\}_{t = 1}^T$ (with $\phi^t \in \mbbR$ and $\lambda^t \in \mbbR_{++}$), such that, at all $t, s$, $\phi^s \leqslant \phi^t + \lambda^t a^{t, s}$.

$(3) \implies (4)$: The construction argument used here follows \cite{afriat1973}. For any $\mbx$ define $U(\mbx) = \min_t \{\phi^t + \lambda^t \mbp^t \cdot (\mbx - e^t \mbx^t)\}$, and notice that $U$ is strictly increasing, continuous, and concave. First, we must have $U(\mbx^t) \geq \phi^t$. To see this, suppose that $U(\mbx^t) < \phi^t$; then $U(\mbx^t) = \phi^m + \lambda^m \mbp^m \cdot (\mbx^t - e^m \mbx^m) < \phi^t$, contradicting (3). Second, for any $\mbx$ satisfying $\mbp^t \cdot \mbx \leq e^t \mbp^t \cdot \mbx^t$, we have $\phi^t \geq \phi^t + \lambda^t \mbp^t \cdot (\mbx - e^t \mbx^t)$ since $\lambda^t > 0$, and $U(\mbx) \leq \phi^t + \lambda^t \mbp^t \cdot (\mbx - e^t \mbx^t)$ by the definition of $U$. Taken altogether,
\[
U(\mbx^t) \geq \phi^t \geq \phi^t + \lambda^t \mbp^t \cdot (\mbx - e^t \mbx^t) \geq U(\mbx).
\]

$(4) \implies (1)$: This is obvious since every well-behaved utility function is (by definition) strictly increasing and thus locally nonsatiated. \hfill \qed

\linespread{1.0} \selectfont \small

\bibliography{rpbib}

\end{document}